\documentclass{article}

\usepackage{amsmath, amsthm, amssymb}
\usepackage{graphicx}

\newtheorem{lemma}{Lemma}
\newtheorem{proposition}{Proposition}
\newtheorem{theorem}{Theorem}
\newtheorem{corollary}{Corollary}

\theoremstyle{definition}
\newtheorem{definition}{Definition}

\newcommand{\shuffle}{\ensuremath{\rotatebox{270}{$\!\!\exists$}}}

\newcommand{\SES}{{\tt SubwordEqSolvability}}
\newcommand{\SIT}{{\tt SubwordIneqAbsoluteness}}

\title{Absoluteness of subword inequality is undecidable}
\author{Shinnosuke Seki}

\begin{document}

\maketitle

\begin{abstract}
	Mateescu, Salomaa, and Yu asked: is it decidable whether a given subword history assumes only non-negative values for all words over a given alphabet. 
	In this paper, we solve this open problem by proving that this problem is undecidable even under stronger conditions than supposed originally. 
\end{abstract}

	\section{Subword history and inequality}
	\label{sec:SH}

Let $\Sigma$ be an alphabet, and by $\Sigma^*$ we denote the set of all words over $\Sigma$ including the empty word $\lambda$. 

{\it Parikh mappings (vectors)}, introduced in \cite{Parikh66}, provide us with numerical properties of a word and a set of words.  
Some specific ordering of the letters in $\Sigma = \{a_1, \ldots, a_n\}$ in mind, the Parikh mapping of a word $w$ is $(|w|_{a_1}, |w|_{a_2}, \ldots, |w|_{a_n})$, where $|w|_a$ denotes the number of occurrences of a letter $a \in \Sigma$ in a word $w \in \Sigma^*$ (for instance, $|aab|_a = 2$ and $|aab|_b = 1$). 
This idea can be generalized as counting in $w$ the number of occurrences of another word $u$ as a (continuous) subword or a scattered subword. 
The latter is of especial interest. 
In general, $u$ is a {\it scattered subword} of $w$ if there exist an integer $k \ge 1$ and words $x_1, \ldots, x_k, y_0, y_1, \ldots, y_k$, some of which are possibly empty, such that 
\[
	\mbox{$u = x_1 \cdots x_k$ and $w = y_0 x_1 y_1 \cdots x_k y_k$}.
\]
For various usages of terminologies, the reader is referred to \cite{RozenbergSalomaa97}. 
Then we can generalize the notation $|w|_a$ as $|w|_u$ to denote the number of occurrences of $u$ as a scattered subword of $w$. 
For instance, $|aab|_{ab} = 2$ because two occurrences of $a$ precede that of $b$. 
It is a convention made in \cite{MaSaYu04} to assume that $|w|_{\lambda} = 1$ for the empty word $\lambda$ and any word $w \in \Sigma^*$. 

The number of scattered subwords can provide more information about the word $w$ itself than Parikh mapping. 
For $\Sigma = \{a, b\}$, the Parikh mapping $(3, 3)$ admits all 20 words in $aaa \shuffle bbb$ like $ababba$ as $w$, where $\shuffle$ is the shuffle operation. 
Adding a condition $|w|_{ab} = 8$ to this Parikh mapping reduces the candidate of $w$ to $aababb$ \cite{MaSaYu04}. 
More advanced logic can be implemented by adding and/or multiplying such conditions; $|w|_a \times |w|_b = 4$ implies that $w \in a \shuffle bbbb \cup aa \shuffle bb \cup aaaa \shuffle b$. 
This idea led Mateescu, Salomaa, and Yu to propose the notion of subword history as follows. 

\begin{definition}[\cite{MaSaYu04}]\label{def:subword_history}
	A {\emph subword history} in $\Sigma$ and its {\it value} in a word $w$ are defined recursively as follows: 
	\begin{itemize}
	\item	Every word $u$ in $\Sigma^*$ is a subword history in $\Sigma$, referred to as {\it monomial}, and its value in $w$ equals $|w|_u$. 
	\item	Assume that $SH_1$ and $SH_2$ are subword histories with values $\alpha_1$ and $\alpha_2$, respectively. 
	Then
	\[
		\mbox{$-(SH_1)$, $(SH_1) + (SH_2)$, and $(SH_1) \times (SH_2)$}
	\]
	are subword histories with respective values
	\[
		\mbox{$-\alpha_1$, $\alpha_1+\alpha_2$, and $\alpha_1 \alpha_2$}. 
	\]
	\end{itemize}
\end{definition}

\noindent
The notation $|w|_u$ is now further generalized as $|w|_{SH}$ for a subword history $SH$ to denote the value of $SH$ in $w$. 

For a non-negative integer $e \ge 0$ and a subword history $SH$, we denote $\overbrace{SH \times SH \times \cdots \times SH}^{\mbox{$e$ times}}$ by $\prod^e SH$. 
For instance, $\prod^2 SH = SH \times SH$ and $\prod^3 SH = SH \times SH \times SH$. 
Let us set $\prod^0 SH$ be $\lambda$ for any subword history $SH$. 
In light of the next proposition, this setting does not contradict the convention that $|w|_{\lambda} = 1$ for any word $w$. 

\begin{proposition}\label{prop:extended_SH}
	Let $SH$ be a subword history in $\Sigma$ with value $\alpha$, $c$ be an integer, and $e$ be a non-negative integer. 
	Then $c(SH)$ and $\prod^e SH$ are subword histories with respective values $c\alpha$ and $\alpha^e$. 
\end{proposition}

Two subwords $SH_1$ and $SH_2$ are {\it equivalent} if $|w|_{SH_1} = |w|_{SH_2}$ for every word $w \in \Sigma^*$. 
It is not difficult to observe that the subword histories $a \times b$ and $ab + ba$ assume the same value in any word (see \cite{MaSaYu04}). 
These two subword histories are hence equivalent. 
A subword history is {\it linear} if it is obtained without using the operation $\times$. 
We say that a linear subword history is {\it of degree $n$} if its longest monomial is of length $n$. 
For instance, the degree of $abb + 2c + 3$ is 3 due to its first term. 
More generally, we can define the degree of a subword history as the minimum degree of equivalent linear subword histories. 

Mateescu, Salomaa, and Yu proposed a method of constructing from a given subword history an equivalent {\it linear} subword history, and as its corollary, the problem of deciding the equivalence of two given subword histories turned out to be decidable \cite{MaSaYu04}. 
In the paper, the authors called for a continuation of research on inequalities between subword histories. 
Specifically, they left the following problem open: for a given subword history $SH$, is it decidable whether $|w|_{SH} \ge 0$ holds for every word $w$ in $\Sigma^*$. 
Let us call this problem \SIT. 
From the point of view of decidability, it is irrelevant whether this problem is formalized with $\ge$ or with $>$. 
Indeed, deciding whether $|w|_{SH} > 0$ holds for every word $w \in \Sigma^*$ is equivalent to deciding whether $|w|_{SH-\lambda} \ge 0$; note that $SH-\lambda$ is a valid subword history with value $|w|_{SH}-1$.

	\section{Main results}

In this section, we prove that $\SIT$ is undecidable even under strong restrictions (Corollary~\ref{cor:SES_SIT_undec_LRSH_9}). 
This is our main contribution in this paper. 

First of all, we show that this problem is at least as hard as the problem of deciding for given two subword histories $SH_1$ and $SH_2$ whether there exists a word $w \in \Sigma^*$ such that $|w|_{SH_1} = |w|_{SH_2}$ holds.
Let us call the latter problem \SES. 
The reader can consult \cite{RozenbergSalomaa94, AroraBarak09}, if needs arise, on undecidability, polynomial-time Karp reduction, and NP-hardness. 

\begin{lemma}\label{lem:Eq_Ineq_Turing_reduction}
	\SES \ is polynomial-time Karp reducible to \SIT. 
\end{lemma}
\begin{proof}
	Assume that as an instance of \SES \ two subword histories $SH_1$ and $SH_2$ are given. 
	Let $SH = SH_1 - SH_2$. 
	Then, the answer to this instance is no if and only if $|w|_{SH \times SH} > 0$ for every word $w \in \Sigma^*$. 
	Note that $SH \times SH$ is a valid subword history (Proposition~\ref{prop:extended_SH}), and its value in $w$ is $(|w|_{SH})^2$. 
\end{proof}

In order to prove the undecidability of $\SIT$, therefore, it suffices to prove that $\SES$ is undecidable. 

\begin{theorem}\label{thm:SES_undec}
	$\SES$ is undecidable. 
\end{theorem}
\begin{proof}
	This proof is based on the unsolvability of Diophantine equation proved by Matiyasevich in \cite{Matiyasevich70}, the answer to the Hilbert's tenth problem. 
	Let a Diophantine equation 
	\[
		\sum_{1 \le i \le \ell} c_i x_1^{e_{i, 1}} x_2^{e_{i, 2}} \cdots x_m^{e_{i, m}} = 0
	\]
	be given, where $\ell \ge 1$, $c_1, \ldots, c_\ell$ are integer constants, $x_1, x_2, \ldots, x_m$ are positive integer variables, and $e_{i, 1}, e_{i, 2}, \ldots, e_{i, m}$ are non-negative integer exponents for $1 \le i \le \ell$. 
	(It is well known that we can restrict the attention to positive integer variables, see \cite{RozenbergSalomaa94}.)

	Let $\Sigma = \{a_1, \ldots, a_m \}$. 
	Consider a word $w$ in $a_1^{n_1} \shuffle a_2^{n_2} \shuffle \cdots \shuffle a_m^{n_m}$ for some non-negative integers $n_1, \ldots, n_m$. 
	Then for $1 \le j \le m$, we have
	\[
		|w|_{a_j} = n_j. 
	\]
	Proposition~\ref{prop:extended_SH} implies that $\prod^{e_{i,j}} a_j$ is a subword history for any $1 \le i \le \ell$, and its value in $w$ is $n_j^{e_{i,j}}$. 
	Using the proposition once again, we see that 
	\[
		c_i \biggl(\prod^{e_{i, 1}} a_1 \times \prod^{e_{i, 2}} a_2 \times \cdots \times \prod^{e_{i, m}} a_m \biggr)
	\]
	is a subword history whose value in $w$ is $c_i n_1^{e_{i, 1}} n_2^{e_{i, 2}} \cdots n_m^{e_{i, m}}$. 
	Let us denote this subword history by $SH_i$, and let $SH = \sum_{1 \le i \le \ell} SH_i$, which is also a subword history.  
	Now it should be clear that 
	\begin{equation}\label{eq:SH_Diophantine_eq}
		|w|_{SH} = \sum_{1 \le i \le \ell} c_i n_1^{e_{i, 1}} n_2^{e_{i, 2}} \cdots n_m^{e_{i, m}}. 
	\end{equation}
	This is the value we can obtain by substituting $(n_1, \ldots, n_m)$ into the given Diophantine equation. 
	Therefore, if the Diophantine equation has a positive integer solution $(n_1, n_2, \ldots, n_m)$, then for such $w$, $|w|_{SH} = 0$. 
	Conversely, assume that there exists a word $v \in \Sigma^*$ such that $|v|_{SH} = 0$. 
	According to Definition~\ref{def:subword_history} and Eq.~(\ref{eq:SH_Diophantine_eq}), 
	\[
		|v|_{SH} = \sum_{1 \le i \le \ell} c_i |v|_{a_1}^{e_{i,1}} |v|_{a_2}^{e_{i,2}} \cdots |v|_{a_m}^{e_{i, m}}. 
	\]
	Since this value is 0, $(|v|_{a_1}, |v|_{a_2}, \ldots, |v|_{a_m})$ is a positive integer solution to the given Diophantine equation. 
	Consequently, if $\SES$ were decidable, then we would be able to determine the solvability of the given Diophantine equation, a contradiction.  
\end{proof}

\begin{corollary}\label{cor:SIT_undec}
	$\SIT$ is undecidable. 
\end{corollary}

As being mentioned previously, the equivalence of two subword histories is decidable, and note that this does not contradict Corollary~\ref{cor:SIT_undec}. 

It is worth observing that in the proof of Theorem~\ref{thm:SES_undec}, we reduce a given Diophantine equation into an element of a restricted class of subword histories, which we call the class of {\it letter-restricted subword histories}. 
The definition of letter-restricted subword history is obtained by restricting that monomials be letters in $\Sigma$ or $\lambda$ in Definition~\ref{def:subword_history}. 

It is well known that the Diophantine equations remain unsolvable even if the number of variables involved is limited to be 9 \cite{Matiyasevich93}. 
In the proof of Theorem~\ref{thm:SES_undec}, the number of variables equals that of letters in $\Sigma$. 
Thus, over an alphabet of 9 letters, $\SES$ is undecidable, and hence, so is $\SIT$. 
Combining this with what was mentioned in the last paragraph, now we present our strongest result on the undecidability of these problems as of this moment. 

\begin{theorem}\label{thm:SES_SIT_undec_LRSH_n}
	If the Diophantine equations are unsolvable over $n$ variables, then $\SES$ and $\SIT$ are unsolvable even for the class of letter-restricted subword histories over an alphabet of $n$ letters. 
\end{theorem}

\begin{corollary}\label{cor:SES_SIT_undec_LRSH_9}
	$\SES$ and $\SIT$ are undecidable even for the class of letter-restricted subword histories over a nonary alphabet. 
\end{corollary}

Corollary~\ref{cor:SES_SIT_undec_LRSH_9} does not mean that $\SES$ or $\SIT$ is decidable over an alphabet of size at most 8. 
It is conjectured that Diophantine equations remain unsolvable even over three variables. 
If so, then Theorem~\ref{thm:SES_SIT_undec_LRSH_n} implies that these problems would be undecidable even for the class of letter-restricted subword histories over a ternary alphabet. 

How small do we have to make the size of alphabet to make these problems decidable? 
We cannot help but leave this matter unsettled in this paper, but can provide a result to illustrate how hard $\SIT$ is. 
Manders proved that it is NP-complete to decide the solvability of a given Diophantine equation of the form $c_1 x^2 + c_2 y + c_3 = 0$ \cite{Manders80}. 
Our construction of a subword history from a given Diophantine equation in the proof of Theorem~\ref{thm:SES_undec} can be done in a polynomial time. 
In addition, the subword history thus constructed can be transformed in a polynomial time into a linear subword history $c_1 a + 2c_1 aa + c_2 b + c_3$ of degree 2 by the above-mentioned product elimination by Mateescu, Salomaa, and Yu, where the letters $a$ and $b$ correspond to the variables $x$ and $y$, respectively. 
With Lemma~\ref{lem:Eq_Ineq_Turing_reduction}, we can prove the following theorem, though it does not settle the question at the beginning of this paragraph.  

\begin{theorem}\label{thm:SES_SIT_NP-hardness}
	$\SES$ and $\SIT$ are NP-hard even for the class of letter-restricted subword histories of degree 2 over a binary alphabet. 
\end{theorem}

	\section{System of Diophantine equations}

In this section, we glance at the polynomial-time Karp reduction from a given system of Diophantine equations to a subword inequality. 
The reduction should be in itself trivial from our proof of Theorem~\ref{thm:SES_undec}, but let us spend some space for this because of an implication it has on a significant problem called {\it preimage problem}. 

A system of Diophantine equations is a finite collection $(Eq_1, Eq_2, \ldots, Eq_k)$ of Diophantine equations. 
Using our method, the equations $Eq_1, \ldots, Eq_k$ are transformed into the respective subword histories $SH_1, \ldots, SH_k$. 
From them, we construct the following subword history: 
\[
	SH := \prod_{1 \le i \le k} ((SH_i \times SH_i) + 1). 
\]
Then, for $w \in \Sigma^*$, $|w|_{SH} = 1$ if and only if for all $1 \le i \le k$, $|w|_{SH_i} = 0$. 
Since $SH$ always assumes a positive integer value, deciding whether $|w|_{SH} = 1$ can be done both by equation and by inequality. 

Given a subword history $SH$ and a word $w \in \Sigma^*$, it is a pen-and-paper calculation to obtain the value of $SH$ in $w$, and it remains the case no matter how many subword histories are given. 
From the subword histories $SH_1, \ldots, SH_k$ and the values $n_1, \ldots, n_k$ thus calculated from $w$, we can build the following system of subword equations: 
\[
	\left\{
	\begin{array}{lcl}
	SH_1 &=& n_1 \\
	&\vdots& \\
	SH_k &=& n_k
	\end{array}
	\right.
\]
and after that, we hide $w$. 
Can we find $w$, or more desirably, can we eliminate the candidates of $w$? 
In Section~\ref{sec:SH}, an example was cited from \cite{MaSaYu04} to see the uniqueness of the word $w \in \{a, b\}^*$ satisfying $|w|_a = |w|_b = 3$ and $|w|_{ab} = 8$. 
This is interpreted in the above-mentioned framework as finding $w$ when $(3, 3, 8)$ is given (assume that we know to what subword history each coordinate is related in this vector). 
Problems of this type are collectively termed {\it preimage problems} (see, e.g., \cite{AkutsuFukagawa05, FuWaZhNaAk08} for a preimage problem in chemoinformatics). 
Preimage problems can be formalized not only for words but for various objects like graphs so long as some of their properties can be quantified. 
What we mentioned in the previous paragraph, however, demonstrates how computationally-hard the preimage problem is even for words. 
One reason for this hardness is that in counting occurrences of a subword, search range covers the whole of a given word (global scope) in our current formalization. 
As such, if we reformulated problems so as to confine the search range, then the reformulated preimage problem could be solved even efficiently. 
In \cite{AkutsuFukagawa05}, Akutsu and Fukagawa counted only the occurrences of words as a {\it continuous} subword, and showed that in this setting, the preimage problem can be solved in a polynomial time. 

	\section{Concluding remarks, discussions, and future directions}

In this paper, it was proved to be undecidable whether there exists a word in which an equation between given two subword histories holds. 
With the polynomial-time Karp reduction, this amounted to the answer to the open problem by Mateescu, Salomaa, and Yu posed in \cite{MaSaYu04}. 
This problem was proved to remain undecidable even under conditions on the size of alphabet, on the class of subword histories considered, and on the length of monomials involved. 
As such, our main results are stronger than a solution to the original open problem. 

Results in this paper are oriented toward unsolvability, and therefore, cannot be employed to make use of a number of known decidability results on the solvability of Diophantine equations (see \cite{Baker84}). 
This motivates us to make a research on the characterization of subword histories that is polynomial-time Karp reducible to a Diophantine equation whose solvability is decidable. 
It might be worth recalling that the Diophantine equations are reduced to the {\it very restricted} class $\mathcal{SH}$ of letter-restricted subword histories. 
Thus, for any class of subword histories that does not contain $\mathcal{SH}$ as a subset, it remains unknown whether $\SES$ or $\SIT$ is decidable. 
The most significant difference between Diophantine equations and equations on subword histories is that the latter is defined on the group which is not Abelian. 
In this paper, this difference has been barely encountered because our attention was mainly on the class of letter-restricted subword histories, in which commutativity does not count so much. 
This observation gives us an impression that combinatorics on words will play an important role in working on the above-mentioned problems (see \cite{Fazekas08} and the references therein). 

	\section*{Acknowledgements}

We wish to express our sincere gratitude to Professor Dr.~Arto Salomaa for reading earlier versions of this paper carefully and making quite a few valuable comments. 
He originated the idea of letter-restricted subword history. 
We are indebted to Dr.~Szilard Zsolt Fazekas for his introducing the open problem to us and for helpful discussions with him. 
This paper also benefited from my discussions with Mr.~Kei Taneishi on Diophantine equations, number theory, and chemoinformatics. 

This research is financially supported by the Funding Program for Next Generation World-Leading Researchers (NEXT program) to Professor Dr.~Yasushi Okuno.



\end{document}